\documentclass{article}

\usepackage{arxiv}
\usepackage[utf8]{inputenc} 
\usepackage[T1]{fontenc}    
\usepackage{hyperref}       
\usepackage{url}            
\usepackage{booktabs}       
\usepackage{amsfonts}       
\usepackage{nicefrac}       
\usepackage{microtype}      
\usepackage{lipsum}
\usepackage{graphicx}
\graphicspath{ {./images/} }
\usepackage{graphicx}%
\usepackage{multirow}%
\usepackage{amsmath,amssymb,amsfonts,bm}%
\usepackage{amsthm}%
\usepackage{mathrsfs}%
\usepackage[title]{appendix}%
\usepackage{xcolor}
\definecolor{myred}{RGB}{150,0,0}

\usepackage{textcomp}%
\usepackage{float}
\usepackage[numbers]{natbib}
\usepackage{manyfoot}%
\usepackage{booktabs}%
\usepackage{algorithm}%
\usepackage{algorithmicx}%
\usepackage{algpseudocode}%
\usepackage{listings}%
\usepackage{tabularx}
\usepackage{circuitikz}
\usepackage{tikz}
\usetikzlibrary{circuits, graphs, positioning}
\usepackage{enumitem}
\hypersetup{hidelinks}
\usepackage{subcaption}

\usepackage{extarrows}

\newtheorem{theorem}{Theorem}
\newtheorem{proposition}{Proposition}
\newenvironment{restatedproposition}[1]{%
  \begingroup
  \renewcommand{\theproposition}{\ref{#1}}
  \begin{proposition}[Restated]%
}{%
  \end{proposition}%
  \addtocounter{proposition}{-1}
  \endgroup
}
\theoremstyle{thmstyletwo}%
\newtheorem{example}{Example}%
\newtheorem{lemma}{Lemma}%
\theoremstyle{thmstylethree}%
\newtheorem{definition}{Definition}%
\title{Polynomial-Time Algorithms for Setting Tight Big-M Coefficients in Transmission Expansion Planning with Disconnected Buses}

\author{
 Behnam Jabbari-Marand \\
  Edward P. Fitts Department of Industrial and Systems Engineering\\
  North Carolina State University\\
  Raleigh, NC 27606, USA \\
  \texttt{Bjabbar@ncsu.edu} \\
   \And
 Adolfo R. Escobedo \\
  Edward P. Fitts Department of Industrial and Systems Engineering\\
  North Carolina State University\\
  Raleigh, NC 27606, USA \\
  \texttt{Arescobedo@ncsu.edu}
}
\makeatletter
\def\@date{}
\makeatother
\begin{document}
\maketitle
\begin{abstract}
The increasing penetration of renewable energy and rising electricity demand are driving the need to integrate new buses and transmission lines into transmission grids. These trends are reshaping transmission expansion planning (TEP), motivating the development of effective methodologies to manage the resulting complexity. This paper introduces the longest shortest-path connection (LSPC) algorithm, a graph-based method to enhance the mixed-integer linear programming disjunctive formulation of TEP using valid inequalities (VIs). Traditional approaches for determining big-M coefficients in disconnected TEP networks typically rely on solving the computationally intensive longest path problem (LPP). In contrast, LSPC circumvents these limitations by efficiently identifying relevant power-flow paths between disconnected buses within the expansion network. We demonstrate that the VIs generated from these identified paths dominate those derived from LPP-based methods and other existing approaches.
\end{abstract}


\section{Introduction}\label{sec1}
Transmission expansion planning (TEP) entails adding transmission lines within and between systems to accommodate future demand growth at the lowest possible cost \cite{garver1970transmission}. The strategic importance of TEP in power systems cannot be overstated, given its long-term implications for system operations. The evolving energy landscape, marked by renewable integration, large-scale generation projects, and increasing demand from industrial facilities and data centers, has significantly increased the complexity of TEP and the need for more effective solution methods \cite{lumbreras2016new,wan2025grid}.

TEP is often formulated using a linear approximation model, namely Direct-Current Optimal Power Flow (DC-OPF), which is widely applied in power system optimization (e.g., \cite{pan2020deepopf},\cite{kargarian2016toward},\cite{minot2016parallel}). This linearization is obtained by assuming uniform voltage magnitudes, minimal angle differences, and disregarding reactive power, considering the low conductance of transmission lines \cite{horsch2018linear}. These simplifications provide an effective trade-off between simplicity and accuracy, making them suitable for TEP, where operational considerations are less critical due to the long-term planning horizon and extensive power transmission distances \cite{lumbreras2014automatic}.

The introduction of discrete decisions to DC-OPF transforms it from a linear program (LP) into a mixed-integer linear program (MILP), which is generally intractable (i.e., NP-hard \cite{oertel2014complexity}). Its complexity explains the emphasis on metaheuristics and other inexact methods (e.g., \cite{dong2025transmission},\cite{romero2005constructive},\cite{de2005transmission},\cite{abdi2021metaheuristics}), even though these approaches do not provide formal guarantees of the solution quality. Hybrid methods for DC-TEP that couple heuristics with branch-and-bound algorithms (e.g., \cite{sahraei2014performance},\cite{gopalakrishnan2012global},\cite{sousa2011heuristic}) have been explored. However, despite having theoretical guarantees, their high computational cost limits them to smaller problems, hindering their practical scalability. Benders' decomposition is another popular method for solving DC-TEP that guarantees optimal solutions (\cite{binato2001new},\cite{mohammadi2013benders},\cite{haffner2000branch},\cite{megel2016reducing}), but it suffers from slow convergence in large-scale problems.
 
A complementary exact approach is to strengthen the MILP formulation. DC-TEP is commonly formulated using big-$M$ constraints that couple line-investment decisions with physical constraints across the incident buses. The choice of the big-$M$ bounds can substantially affect formulation strength, as overly large values weaken the LP relaxation and increase the computational burden of exact methods. Cutting plane methods can tighten these bounds. Their computational benefits have been widely demonstrated in power systems optimization problems with discrete decisions, such as optimal transmission switching and unit commitment (e.g., \cite{dey2022node,kocuk2016cycle,lorca2016multistage,hedman2010co}). In contrast, they have received little attention in DC-TEP. The few works that attempt to strengthen big-$M$ bounds are either limited to already-connected networks or otherwise lead to computationally intractable procedures (e.g., \cite{skolfield2022derivation}).

Deriving tight big-$M$ bounds for candidate-line endpoints—and, more generally, for any bus pair—relies on identifying relevant power-flow paths that connect them. In an already-connected network, any path comprised of existing lines between the pair provides a valid bound. In contrast, identifying relevant paths becomes more challenging when connectivity between buses depends on which investment options are chosen. Existing approaches generally assume that obtaining a provably tight bound in these situations requires solving the longest path problem (LPP), which is NP-hard \cite{schrijver2003combinatorial}. To address this challenge, this paper exploits structural properties of DC-TEP in restricted yet practically relevant expansion settings to identify the relevant paths and thereby derive the tightest big-$M$ bounds efficiently.

The rest of the paper is organized as follows. Section \ref{sec2} presents the DC-TEP formulation and introduces background concepts. Section \ref{sec3} motivates the featured methodology, denoted as the longest shortest-path connection algorithm, to tackle situations that require new-bus integration. Section \ref{sec4} provides a detailed description of this algorithm, along with proofs of correctness and complexity. Finally, Section~\ref{sec5} concludes with a summary of the contributions and directions for future work.
\section{Modeling framework and background}\label{sec2}
This section provides the notation and underlying mathematical model that serves as the basis for the methodology developed in this work. In power systems, \emph{buses} (i.e., nodes) represent connection points for various electrical components (e.g., power plants, substations) and are linked by \textit{corridors}, which are transmission pathways between buses. This work distinguishes between \textit{established corridors}, which connect buses solely through existing transmission lines, and \textit{expansion corridors}, which incorporate candidate lines for potential network expansion. For simplicity, it is assumed that each corridor can accommodate only a single existing line or a candidate line.
\subsection{Notation overview}\label{subsec21}
\subsubsection*{Sets}
\begin{tabular}{ll}
$n$  $\in \mathcal{B}$ & Buses (i.e., nodes) \\
$(i, j)$  $\in \Omega^0$ &Established corridors: corridors containing only an established line\\
$(i, j)$  $\in \Omega^1$ &Expansion corridors: corridors containing only a candidate line\\
\end{tabular}
\subsubsection*{Parameters}
\begin{tabular}{ll}
$c_{i j}$ & Cost of installing a line in corridor $(i, j)\in\Omega^1$ \\
$c_n$ & Cost per unit of power generation at bus $n$ \\
$\overline{g}_n$ & Upper limit of power generation at bus $n$ \\
$d_n$ & Active power demand at bus $n$ \\
$\overline{\theta}_{i j}$ & Limit on the angle difference between buses $i$ and $j$ \\
$\overline{P}_{i j}$ & Maximum capacity of (candidate) line within corridor $(i, j)\in \Omega^1$\\
$\overline{P}_{i j}^0$ & Maximum capacity of (existing) line within corridor $(i, j)\in \Omega^0$\\
$x_{i j}$ & Reactance of line in corridor $(i, j)$ \\
$\sigma$ & A scaling factor for aligning generation costs with transmission investment costs
\end{tabular}
\subsubsection*{Variables}
\begin{tabular}{ll}
$P_{i j}^0$ & Active power transmitted through the existing line in corridor $(i, j)$ \\
$P_{i j}$ & Active power transmitted through the candidate line in corridor $(i, j)$ \\
$g_n$ & Active power produced by the generator at bus $n$ \\
$\theta_n$ & Voltage angle at bus $n$\\
$y_{i j}$ & $\begin{cases} 
1 & \text{If a candidate line within corridor $(i, j)$ is purchased} \\
0 & \text{Otherwise} 
\end{cases}$ 
\end{tabular}
\subsection{MILP disjunctive formulation of DC-TEP}\label{subsec22}
This paper employs the MILP disjunctive model of DC-TEP \cite{villasana1984transmission} with a single investment period. This formulation is as follows:
\begin{subequations}
\begin{align}
& \min \sum_{(i, j) \in \Omega^1} c_{i j} y_{i j}+\sum_{n \in \mathcal{B}} \sigma c_n g_n && \label{eq:tepobj}\\
&\sum_{(i, n) \in \Omega^0}\hspace{-2pt} P_{i n}^0+ \hspace{-2pt}\sum_{(i, n) \in \Omega^1}\hspace{-2pt} P_{i n} - \hspace{-2pt} \sum_{(n, i) \in \Omega^0}\hspace{-2pt} P_{n i}^0 -\hspace{-2pt} \sum_{(n, i) \in \Omega^1} \hspace{-2pt} P_{n i} +g_n=d_n && \forall n \in \mathcal{B}  \label{eq:cons1}\\
& -\overline{P}_{i j}^0 \leq P_{i j}^0 \leq \overline{P}_{i j}^0 && \forall(i, j)\in \Omega^0\label{eq:cons2} \\
& -\overline{P}_{i j} y_{i j} \leq P_{i j} \leq \overline{P}_{i j} y_{i j} && \forall(i, j)\in \Omega^1 \label{eq:cons3} \\
& x_{i j} P_{i j}^0=\left(\theta_i-\theta_j\right)&& \forall(i, j)\in \Omega^0 \label{eq:cons4} \\
& -\overline{\theta}_{ij}(1-y_{ij}) \leq x_{ij}P_{ij} - (\theta_i - \theta_j) \leq \overline{\theta}_{ij}(1-y_{ij}) && \forall (i, j)\in \Omega^1  \label{eq:cons5}\\
& 0\leq g_n \leq \overline{g}_n && \forall n \in \mathcal{B} \label{eq:cons6}\\
& y_{i j} \in\{0,1\} && \forall(i, j)\in \Omega^1 \label{eq:cons8}
\end{align}
\end{subequations}
Objective function \eqref{eq:tepobj} minimizes the total cost of adding new lines and power generation; generation costs are scaled by a factor of $\sigma$ to ensure comparability to investment costs. Constraints \eqref{eq:cons1} enforce Kirchhoff’s Current Law, also known as the flow balance equations. They ensure that at each bus, the inflow and generation are together equal to the summed outflow and demand. Constraints \eqref{eq:cons2} and \eqref{eq:cons3} are capacity limits for existing and candidate lines, respectively. Constraints \eqref{eq:cons4} and \eqref{eq:cons5} enforce Kirchhoff's Voltage Law (KVL) for existing and candidate lines, respectively, by equating the product of line reactance and power flow to the corresponding bus angle difference within a corridor. Here, the susceptance parameter $b_{ij}$ from the standard DC-OPF formulation is replaced with the reactance $x_{ij}=\frac{-1}{b_{ij}}$ to simplify the notation in the formulations. Constraints \eqref{eq:cons5} employ sufficiently large disjunctive parameters, i.e., $\overline{\theta}_{ij}$, to guarantee inequality redundancy for unconstructed corridors. Previous studies use a big-$M$ parameter for the disjunctive coefficient; this work uses the parameter $\overline{\theta}_{ij}$ to generalize its application to all bus pairs (the standard TEP formulation defines this parameter only for adjacent buses). The remaining constraints specify the domain of values for the decision variables. 
\subsection{Review of TEP formulation improvement approaches}\label{subsec23}
In constraints \eqref{eq:cons5}, the disjunctive parameter $\overline{\theta}_{ij}$ imposes an upper limit on the angle difference between buses connected by expansion corridor $(i,j)\in \Omega^1$. The selection of $\overline{\theta}_{ij}$ significantly contributes to the strength of the problem formulation. These coefficients must be sufficiently large not to cut off any integer-feasible solution. However, they should be kept as small as possible to provide tighter LP relaxations, thereby expediting solution times (and mitigating numerical issues \cite{rahmani2013study}). To elaborate, consider a pair of buses $i$ and $j$, succinctly represented as $[i,j]$, that are connected by an established corridor. A valid upper bound on the angle difference between the pair can be derived by incorporating the line capacity \eqref{eq:cons2} into the KVL constraint \eqref{eq:cons4} as
\begin{equation}
    |\theta_i-\theta_j| \leq \overline{P}^0_{ij}x_{ij}=:w_{ij}, \label{eq:capacity-reactance}
\end{equation}
from the fact that $|P^0_{ij}|\leq \overline{P}^0_{ij}$. Henceforth, we refer to the \textit{capacity-reactance product} of line $(i,j)$ (see the right-hand side of \eqref{eq:capacity-reactance}, $\overline{P}^0_{ij}x_{ij}$) as its weight and denote it by $w_{ij}$.

The above bound is valid when a transmission line is constructed between buses $i$ and $j$, but it may not apply to corridors not already connected via an established corridor. Nonetheless, it is possible to utilize \eqref{eq:capacity-reactance} to establish bounds on buses that can be connected via expansion lines. \citet{tsamasphyrou2000transmission} derive a simple bound, denoted by $\gamma$, that applies to every adjacent bus pair in $\Omega^1$ and is given by
\begin{equation}
    |\theta_i-\theta_j|\leq \sum_{(k,l)\in\Omega^0\cup \Omega^1} w_{kl}=:\gamma,\quad \forall (i,j)\in \Omega^1. \label{eq:naivebound}
\end{equation}
Upper bound $\gamma$ effectively represents a worst-case scenario where any flow traveling between $i$ and $j$ would traverse all the corridors ($\Omega^0\cup\Omega^1$) in the network. It is obtained by sequentially applying the angle-difference inequality \eqref{eq:capacity-reactance} to each corridor and summing the results.

Upper bound $\gamma$ is straightforward to compute, but its magnitude becomes excessive even for very small instances, meaning it does not provide real computational advantages. To derive tighter upper bounds, it is necessary to restrict attention to more relevant power flows between bus pairs. This entails identifying and analyzing only the relevant paths that power flow can take between buses $i$ and $j$, which may either be an \textit{established path}, denoted by $\rho^0_{ij}$, composed of existing lines, or a \textit{candidate path}, denoted by $\rho_{ij}$, consisting of expansion corridors (i.e., candidate lines) with or without established corridors (i.e., existing lines). The general approach involves efficiently identifying such paths between bus pairs, and calculating their weights—by summing the line weights along each path (see \eqref{eq:capacity-reactance})—to refine $\overline{\theta}_{ij}$. To help explain the applicability of this approach, let $G=(\mathcal{B},\Omega^0 \cup \Omega^1)$ be the \textit{expansion network} associated with the inclusion of potential investment decisions in TEP. Additionally, it is necessary to define the \textit{initial network} $G^0=(\mathcal{B},\Omega^0)$, which consists solely of existing lines.

Conventional TEP methods model expansion as the addition of individual transmission lines to an already connected network. Figure~\ref{fig:fig1} illustrates this setting, where solid edges denote established corridors, dashed edges denote expansion corridors, and all lines have unit weights. When buses $i$ and $j$ within expansion corridor $(i,j)$ are connected via an established path $\rho^0_{ij}$, a valid upper bound on their voltage-angle difference is obtained by traversing the path. That is, starting from one endpoint of the expansion corridor, $i$, and following path $\rho^0_{ij}$ to the opposite endpoint $j$, summing the inequalities \eqref{eq:capacity-reactance} creates a telescoping effect on the left-hand side, resulting in the angle difference $|\theta_i-\theta_j|$. Simultaneously, the right-hand side accumulates the weights of the traversed lines.
\begin{example}\label{ex:ex1}
    Figure~\ref{fig:fig1}a depicts expansion corridor $(i_1,i_3)$ and established path $\rho^{0}_{i_1i_3}:=\langle(i_1,i_5),(i_5,i_4),(i_4,i_n),(i_n,i_3)\rangle$. Summing the angle-difference inequalities along $\rho^{0}_{i_1i_3}$, represented by the dotted arrows, yields the angle-difference upper bound for $[i_1,i_3]$
\begin{subequations}\label{eq:telescop}
\begin{align} 
  \Big\lvert\sum_{(k,l) \in \rho^{0}_{i_1i_3}}\hspace{-5pt} (\theta_{k} - \theta_{l})\Big\rvert\leq & \;
\underbrace{|\theta_{i_1} - \theta_{i_5}|}_{\substack{\leq w_{i_1i_5}}} +  
\underbrace{|\theta_{i_5} - \theta_{i_4}|}_{\substack{\leq w_{i_5i_4}}} +  
\underbrace{|\theta_{i_4} - \theta_{i_n}|}_{\substack{\leq w_{i_4i_n}}} +  
\underbrace{|\theta_{i_n} - \theta_{i_3}|}_{\substack{\leq w_{i_ni_3}}}\Rightarrow & \\
|\theta_{i_1}-\theta_{i_3}| \leq & \; w_{i_1i_5} + w_{i_5i_4} +  w_{i_4i_n} + w_{i_ni_3} = 1+1+1+1 = 4\,. & \label{eq:angineq}
  \end{align} 
  \end{subequations}
\end{example}
\begin{figure}
    \centering
    \includegraphics[width=0.8\linewidth]{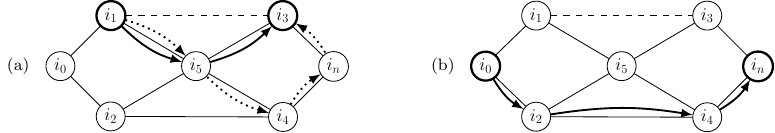}
    \caption{Identifying relevant paths between bus pairs based on their connectivity in $G^0$ and $G$: (a) Bus pair $[i_1,i_3]$ is adjacent in $G$ and connected in $G^0$; (b) Bus pair $[i_0,i_n]$ is non-adjacent in $G$ yet connected in $G^0$. Unit line weights are assumed in the diagrams.}
   \label{fig:fig1}
\end{figure}

The presence of multiple established paths connecting $i$ and $j$ results in multiple inequalities similar to \eqref{eq:angineq}. \citet{binato2001new} propose solving a shortest path problem (SPP) within $G^0$ to identify the shortest established path, denoted by $\underline{\rho}^0_{ij}$, which yields the tightest bound on $\overline{\theta}_{ij}$. In Example \ref{ex:ex1}, the shortest path $\underline{\rho}^{0}_{i_1i_3}:=\langle(i_1,i_5),(i_5,i_3)\rangle$ (indicated by solid arrows) yields an upper bound of $2$ on the angle difference for buses $i_1$ and $i_3$, dominating the previously determined bound.

\citet{skolfield2022derivation} extend the application of this approach, previously limited to adjacent buses (i.e., those with an expansion corridor between them), to derive angle-difference bounds $\overline{\theta}_{ij}$ for any $i,j\in\mathcal{B}$ that is connected in $G^0$ (i.e., buses that can reach each other via existing lines). The authors solve an SPP between $i$ and $j$ over $G^0$ and set $\overline{\theta}_{ij}$ equal to the total weight of the resulting shortest path, $w(\underline{\rho}^0_{ij})$. Figure~\ref{fig:fig1}b illustrates this for non-adjacent buses $i_0$ and $i_n$ connected in $G^0$, where the shortest path between them yields an angle-difference upper bound for $[i_0,i_n]$. This broader applicability enables the derivation of additional angle-difference bounds beyond those proposed by \citet{binato2001new}, which can be exploited to strengthen the formulation.

Existing techniques for deriving tight big-$M$ bounds rely on established paths that connect the bus pair. However, these methods are not applicable in expansion situations that incorporate prospective subnetworks of new buses. To motivate this limitation, we next describe the underlying network-design problem. Figure~\ref{fig:fig6} illustrates this problem, with prospective subnetworks shown in bold.

Deriving a tight big-$M$ bound for a bus pair is tantamount to identifying a path whose weight cannot be exceeded under any plausible network configuration; that is, its weight must be no smaller than any shortest path that could connect the pair once the investments are realized. As shown in Figure~\ref{fig:fig6}, identifying such a path can be entirely decision-dependent for certain node pairs. For example, for nodes $m_0$ and $i_1$, no established path is available, so a tight bound must be derived from candidate paths. Restricting attention to the prospective subnetwork containing $m_0$, there are already eight candidate paths from $m_0$ to $i_1$. There are also eight candidate paths from $m_0$ to $i_2$, which can combine with paths through the existing grid and other prospective subnetworks to reach $i_1$. This multiplicative growth leads to a combinatorial explosion in real-size networks. Hence, the general network-design problem is intractable.
\begin{figure}
    \centering
    \includegraphics[width=0.5\linewidth]{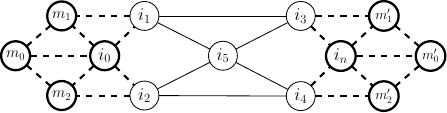}
    \caption{General network-design setting with prospective subnetworks of new buses, shown as bold nodes, where connectivity between some bus pairs depends heavily on expansion decisions.}
    \label{fig:fig6}
\end{figure}
\section{Importance of the longest shortest-path}\label{sec3}
The combinatorial difficulty of tightening big-$M$ bounds for disconnected bus pairs (e.g., $m_0$ and $m^\prime_0$ in Figure~\ref{fig:fig6}) stems from the complex (i.e., meshed) structure of the new-bus subnetworks. To ensure that no feasible solution is excluded, \citet{binato2001new} therefore propose a complete-enumeration approach that evaluates all simple candidate paths in $G$ between a disconnected bus pair $[i,j]$ and selects the longest one, denoted by $\overline{\rho}_{ij}$. This is equivalent to solving an LPP, which is NP-hard \cite{schrijver2003combinatorial}.

\citet{skolfield2022derivation} introduce a more sophisticated approach that leverages parallel paths between the bus pair to further tighten the angle-difference bound obtained with complete enumeration. However, in a disconnected setting, this approach still requires solving the LPP, since an initial bound must be available before it can be further tightened. Beyond its computational difficulty, the total weight of an LPP-based path often exceeds that of plausible power-flow paths between buses, especially in large meshed grids, leading to overly conservative angle-difference bounds (see the next section for numerical examples and theoretical guarantees).

Next, we explain how the LPP can often be circumvented in restricted, yet practical expansion situations where the initial network is connected and the new buses form relatively simple prospective subnetworks. These subnetworks may represent remote, geographically dispersed resources and substations \cite{madrigal2012transmission}. Unlike the general setting in Figure~\ref{fig:fig6}, such subnetworks are often radial, forming \emph{expansion trees} that connect to the existing grid through a small number of candidate lines \cite{ACEG2023ReadyToGo}.
\begin{definition}
    An \emph{expansion tree} $T(\mathcal{B}_T, \Omega^1_T)$ is an acyclic, connected subgraph of $G$ rooted at a bus that is separated from $G^0$ by a single expansion corridor.
\end{definition}
Figure~\ref{fig:fig0} provides a stylized representation of this structure, highlighting the expansion trees in bold. In this setting, \emph{tie lines}, i.e., candidate lines incident to the \emph{roots} of expansion trees (e.g., $i_0$ and $i_n$), enable the connection of new buses in these trees to the existing grid. This structure, together with the connectivity of the existing network, can be exploited to substantially narrow the set of relevant candidate paths.

\begin{figure}[t]
    \centering
    \includegraphics[width=0.52\linewidth]{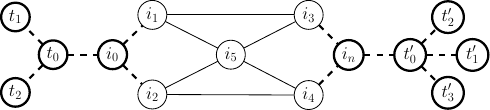}
    \caption{Illustration of a transmission expansion setting with radial prospective subnetworks. The expansion trees rooted at \(i_0\) and \(i_n\), namely \(\{(i_0,t_0),(t_0,t_1),(t_0,t_2)\}\) and \(\{(i_n,t'_0),(t'_0,t'_1),(t'_0,t'_2),(t'_0,t'_3)\}\), can be linked to the existing grid, i.e., \(\{i_1,\ldots,i_5\}\), through candidate lines \(\{(i_0,i_1),(i_0,i_2)\}\) and \(\{(i_3,i_n),(i_4,i_n)\}\), respectively.}
    \label{fig:fig0}
\end{figure}
To motivate this insight, it is necessary to formally define valid inequalities (VIs). For an integer programming problem, written succinctly as $\min\{cx: x \in \mathcal{S}\}$, with feasible region $\mathcal{S}=\{x\in \mathbb{Z}^n: Ax \leq b\}$, an inequality $\pi x \leq \pi_0$ is valid if it holds for all feasible solutions $x\in \mathcal{S}$ \cite{wolsey2014integer}. A VI is deemed effective if it reduces the feasible region of the relaxed problem represented by polyhedral set $\mathcal{P}=\{x\in \mathbb{R}^n: Ax \leq b\}$. For instance, \citet{skolfield2022derivation} introduce \textit{path-based} VIs to tighten the angle-difference bounds derived from SPP or LPP. For candidate paths $\rho_{i_0i_n}$ parallel to the original SPP- or LPP-based path connecting $i_0$ and $i_n$, these VIs can be written as\begin{equation}
|\theta_{i_0}-\theta_{i_n}|\leq w(\rho_{i_0i_n}) + (1-y_{\rho_{i_0i_n}}) (\overline{\theta}_{i_0i_n}-w(\rho_{i_0i_n})),\label{eq:simplevi}
\end{equation}
where $y_{\rho_{i_0i_n}}\in\{0,1\}$ equals $1$ if $\rho_{i_0i_n}$ is built and $0$ otherwise. Without loss of generality, path-based VI \eqref{eq:simplevi} is presented in a form simplified from the original exposition by associating a binary variable with a path; in contrast, the original VI enumerates all corridors along the path, with one binary variable for each expansion corridor. When $y_{\rho_{i_0i_n}}=1$, the tighter angle-difference bound $w(\rho_{i_0i_n})$ is enforced; when $y_{\rho_{i_0i_n}}=0$, the inequality reverts to the initial bound $\overline{\theta}_{i_0i_n}$, which equals $w(\underline{\rho}^0_{i_0i_n})$ if $i_0$ and $i_n$ are connected in $G^0$, and $w(\overline{\rho}_{i_0i_n})$, i.e., the longest path weight, otherwise. Figure~\ref{fig:fig4}a and Figure~\ref{fig:fig4}b show example candidate paths (indicated by the dotted arrows) for deriving path-based VIs when bus pair $[i_0,i_n]$ is connected (left subfigure) and disconnected in $G^0$ (right subfigure). Unit weights are assumed for all lines in the figure.

As a first contribution, we add a condition for ensuring the validity of the path-based inequalities, namely,\vspace{-6pt}
\begin{equation}
   \{\rho_{i_0i_n}\in \mathcal{C}^{G}_{i_0i_n}\mid w(\rho_{i_0i_n})<\overline{\theta}_{i_0i_n}\}, \label{eq:condition}
\end{equation}
where $\mathcal{C}^{G}_{i_0i_n}$ represents the set of all paths connecting $i_0$ to $i_n$ in the network $G$. Violating this condition, which is missing from the original expression by \citet{skolfield2022derivation}, could lead to \eqref{eq:simplevi} generating invalid inequalities. Proposition \ref{theorem:theo1} in Appendix \ref{secA1} presents the complete expression of inequality \eqref{eq:simplevi} (i.e., using expansion corridor variables) and demonstrates that it is valid only for candidate paths that meet condition \eqref{eq:condition}. In addition, Proposition \ref{prop:prop2} in Section \ref{sec4} shows that, for a candidate path $\rho_{i_0 i_n}$, the path-based VI constructed from an LPP-based initial bound is dominated by the corresponding VI initialized with the bounds developed in the next section. This proposition indicates that the effectiveness of path-based VIs diminishes as the initial angle-difference bound $\overline{\theta}_{i_0i_n}$ grows, with excessively large bounds failing to tighten the LP relaxation.

This observation underscores the significance of identifying shorter relevant paths between initially disconnected bus pairs. Indeed, LPP is unnecessary and ineffective when the existing network between two disconnected buses has a high degree of connectivity. In such situations, shorter connections can often be identified. In practice, very few buses are disconnected; hence, the number of potential corridor combinations for connecting them is small. This is supported by empirical evidence showing that, in large-scale power systems, most buses have relatively few incident lines. Namely, the degree of a bus $i$, denoted by $\deg(i)$, averages about $2.5$ \cite{albert2004structural}. The following example illustrates this point.
\begin{figure}
    \centering
    \includegraphics[width=0.8\linewidth]{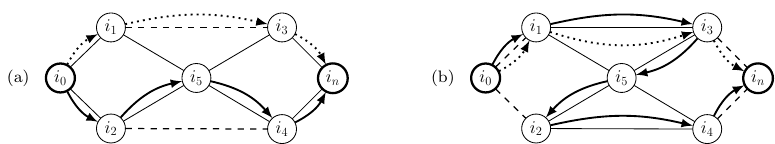}
    \caption{Tightening SPP- and LPP-based angle-difference bounds using path-based VIs: (a) Bus pair $[i_0,i_n]$ is connected in $G^0$ and a candidate path (dotted arrows) shorter than the SPP-based path exists; (b) Bus pair $[i_0,i_n]$ is disconnected in $G^0$ and a candidate path (dotted arrows) shorter than the LPP-based path exists.}
    \label{fig:fig4}
\end{figure}
\begin{example}
     Figure~\ref{fig:fig4}b illustrates the longest path between buses $i_0$ and $i_n$, namely $\overline{\rho}_{i_0i_n}:=\langle(i_0,i_1),(i_1,i_3),(i_3,i_5),(i_5,i_2),(i_2,i_4),(i_4,i_n)\rangle$, which is represented by solid arrows and has a total weight of $6$. However, assuming the construction of corridors $(i_0,i_1)$ and $(i_4,i_n)$, a shorter and more relevant path of weight $4$ can be identified: $\rho_{i_0i_n}:=\langle(i_0,i_1),(i_1,i_5),(i_5,i_4),(i_4,i_n)\rangle$. Both paths traverse the terminal corridors $(i_0,i_1)$ and $(i_4,i_n)$, but the latter uses a shorter path through the existing grid to connect the intermediate buses $i_1$ and $i_4$.
\end{example}
Building upon these observations, we introduce the longest shortest-path connection (LSPC) algorithm for deriving tight angle-difference upper bounds for new-bus integration situations. For a given disconnected pair, the algorithm circumvents the prohibitive computational effort needed to solve an LPP by evaluating connections between the disconnected buses and their neighboring buses. These neighboring buses are then linked through the shortest path within $G^0$. By merging the connections with the neighbors and the SPP-based sub-paths between them, complete paths between the original buses are formed. The longest of these shortest-path connections is then selected to establish a valid angle-difference bound.
    
Before proceeding, it is crucial to distinguish the proposed concept from the \textit{diameter} of graph $G(V,E)$, denoted as $diam(G):=\underset{u,v\in V}\max\; d(u,v)$ \cite{merris2011graph}. In words, the diameter represents the maximum shortest path distance between any two nodes in a static graph. This metric fails to guarantee a feasible bound due to the variable nature of the network in expansion planning. When applied to the expansion graph $G$, the diameter may result in overly conservative estimates by considering the bus pair with the longest shortest path among all pairs in the network. While it evaluates shortest paths between all buses, which may or may not be constructed, the LSPC algorithm focuses only on the part of the graph that is relevant to the specific disconnected bus pair.

Deriving an efficient methodology for tightening big-$M$ bounds for disconnected buses is motivated by growing renewable integration, such as ERCOT's CREZ program, which delivers remote wind power to high-demand regions in Texas \cite{kirby2007evaluating,du2023renewable}, and by the need to connect geographically dispersed energy hubs \cite{bastianel2025identification,fadly2019geographical,Nerc2023}. These settings often induce expansion trees of new buses that are disconnected from one another and from buses in the existing grid, as illustrated in Figure~\ref{fig:fig0}. Because line-investment decisions are not known \emph{a priori}, identifying relevant paths between disconnected buses is challenging. LSPC exploits the underlying network structure, namely the connectivity of the existing grid and the small bus degrees, to efficiently derive the tightest angle-difference bound for each disconnected pair.
\section{The longest shortest-path connection algorithm}\label{sec4}
This section introduces the LSPC algorithm, a novel graph-based approach for tightening the formulation of the DC-TEP problem. The algorithm has two phases, which are applicable to different expansion situations. The results from the first phase are used to derive additional bounds for the second phase.
\subsection{LSPC Phase I: Root-bus integration}\label{subsec41}
The LSPC algorithm establishes an angle-difference upper bound for disconnected bus pairs when expansion trees of new buses are connected to an existing network through a small number of candidate tie lines. Phase I of LSPC focuses on the simplest case, common to many real-world situations, where new buses are separated from the existing grid by a single expansion corridor, i.e., a candidate tie line incident to the root. To illustrate, consider bus pair $[i_0,i_n]$ in Figure~\ref{fig:LSPCillustration}, where bold arrows indicate candidate paths connecting $i_0$ and $i_n$. Note that, since the pair is disconnected in $G^0$, solving an SPP between $i_0$ and $i_n$ does not yield a valid angle-difference bound for these two buses.
\begin{figure}[t]
    \centering
    \includegraphics[width=0.8\linewidth]{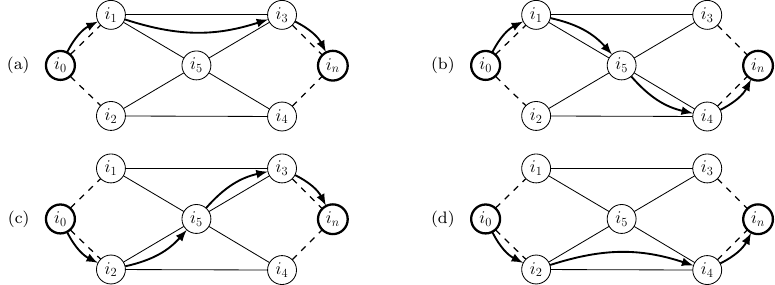}
    \caption{The four shortest-path connections between $i_0$ and $i_n$ obtained by pairing one tie line incident to $i_0$ with one tie line incident to $i_n$: LSPC forms the full paths (a) $\rho^a_{i_0i_n}:=\langle (i_0,i_1),\underline{\rho}^0_{i_1i_3},(i_3,i_n)\rangle$, (b) $\rho^b_{i_0i_n}:=\langle (i_0,i_1),\underline{\rho}^0_{i_1i_4},(i_4,i_n)\rangle$, (c) $\rho^c_{i_0i_n}:=\langle (i_0,i_2),\underline{\rho}^0_{i_2i_3},(i_3,i_n)\rangle$, and (d) $\rho^d_{i_0i_n}:=\langle (i_0,i_2),\underline{\rho}^0_{i_2i_4},(i_4,i_n)\rangle$ and selects the longest. Unit line weights are assumed in the diagrams.\vspace{-10pt}}
    \label{fig:LSPCillustration}
\end{figure}

To proceed, let $N_G(i)$ denote the set of buses adjacent to bus $i$ in $G$, and let $\mathcal{C}^{G}_{ij}$ denote the set of paths connecting $i$ and $j$ in $G$. Phase I applies to a disconnected bus pair $[i,j]$ if every \emph{reachable} neighboring pair $[n_i, n_j]\in N_G(i)\times N_G(j)$, i.e., every pair connected in the expansion grid (i.e., $\mathcal{C}^{G}_{n_in_j}\neq\emptyset$), is also connected in $G^0$ (i.e., $\mathcal{C}^{G^0}_{n_in_j}\neq\emptyset$). This connectivity in $G^0$ ensures that any candidate path between $i$ and $j$ can be represented by selecting candidate lines $(i,n_i)$ and $(n_j,j)$ and linking $n_i$ and $n_j$ through $G^0$. Specifically, a relevant candidate path between $i$ and $j$ is obtained by first identifying the shortest path between $n_i$ and $n_j$ in $G^0$, and then adding the weights of the corresponding terminal corridors, i.e., the tie lines $(i, n_i)$ and $(n_j, j)$, to the weight of the latter sub-path. Since no individual shortest-path connection is guaranteed to be constructed, each induces only a lower bound for $\overline{\theta}_{ij}$. To ensure feasibility, LSPC therefore selects the longest among them; we denote this output path as $\hat{\rho}_{ij}$.
\begin{example}\label{ex:lpplspc}
For bus pair $[i_0,i_n]$ in Figure~\ref{fig:LSPCillustration}, four shortest-path connections arise since $|N_G(i_0)\times N_G(i_n)|=4$. Each subfigure depicts one such path, formed by the starting and ending tie-line segments, with the respective shortest path through the existing network inserted between them. LSPC evaluates these four shortest-path connections and selects the longest among them; here, $\hat{\rho}_{i_0i_n}$ can be either $\rho^b_{i_0i_n}$ or $\rho^c_{i_0i_n}$, each with total weight $4$. Figure~\ref{fig:LPP_LSPC} compares the paths selected by LPP and LSPC. Although both methods select the same tie lines, $(i_0,i_1)$ and $(i_4,i_n)$, LPP connects them through the longest path in $G$, whereas LSPC connects them through the shortest path in $G^0$. Expressly, here LPP selects $\overline{\rho}_{i_0i_n}:=\langle (i_0,i_1),(i_1,i_3),(i_3,i_5),(i_5,i_2),(i_2,i_4),(i_4,i_n)\rangle$, with total weight $6$, exceeding the weight of $\hat{\rho}_{i_0i_n}$ by a factor of $1.5$.
\end{example}
Phase I iterates over all reachable neighboring bus pairs $[n_i, n_j]$ in $G$ to construct all relevant paths between buses $i$ and $j$. Nevertheless, given the typically low degree of buses in a power grid, this usually yields a small number of options. The longest among them is extracted to establish $\overline{\theta}_{ij}$. To proceed, we introduce two concepts:
\vspace{3pt}
\begin{definition}
    A node-isolated graph of $G$ with respect to node $i$ excludes all edges connecting node $i$ to its neighbors, that is, $\overline{G}_i:=(\mathcal{B},\{\Omega^0 \cup \Omega^1\} \backslash \{(i,N_G(i))\})$.
    \end{definition}
    \vspace{3pt}
   Using a \textit{node-isolated} graph ensures that only the reachable neighbors of $i$ and $j$ (i.e., $\mathcal{C}^{\overline{G}_i \cap \overline{G}_j}_{n_in_j}\neq \emptyset$) are considered in identifying candidate paths between $i$ and $j$. Absent this, if  one pair of neighbors is connected by a path in $G$, other neighboring pairs may be incorrectly deemed to be connected by the same path, reaching it through the corridors $(i, N_G(i))$ and $(N_G(j), j)$.
   \vspace{3pt}
  \begin{definition}\label{def:reach}
    The reachability set $\mathcal{C}_{ij}^N$ for buses $i$ and $j$ includes $[n_i,n_j]$, whenever $n_i$ and $n_j$ are connected in ${\overline{G}_i \cap \overline{G}_j}$, and it includes $[i,j]$ if $(i,j)\in \Omega^1$; in mathematical notation, its elements are given by \vspace{-7pt}
    \begin{equation}
        \{[n_i,n_j]\hspace{-2pt}\mid \hspace{-2pt} n_i \in N_G(i), n_j \in N_G(j),n_i\neq j,n_j\neq i, \mathcal{C}_{n_i n_j}^{\overline{G}_i \cap \overline{G}_j} \neq \emptyset\}\cup\{[i,j]\hspace{-2pt}\mid \hspace{-2pt}(i,j)\in \Omega^1\}. \label{eq:reachabilityset} 
    \end{equation}
\end{definition}
\begin{figure}[t]
    \centering
    \includegraphics[width=0.8\linewidth]{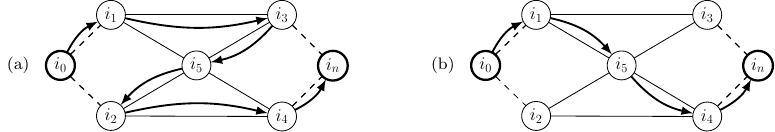}
    \caption{Power-flow paths connecting $[i_0,i_n]$ selected by LPP and LSPC: (a) LPP links the tie-line endpoints through the longest simple path in $G$; (b) LSPC links them through the shortest path in $G^0$. Unit line weights are assumed in the diagrams.}
    \label{fig:LPP_LSPC}
\end{figure}
As noted earlier, each shortest-path connection associated with $[\eta_i,\eta_j]\in \mathcal{C}_{{ij}}^N$, where $\eta_i\in\{i,n_i\}$, provides a lower bound for $\overline{\theta}_{ij}$, since the corresponding path may or may not be built in the TEP solution. Consequently, the longest of these connections yields the tightest attainable value of $\overline{\theta}_{ij}$. Specifically, the set $\mathcal{C}_{{ij}}^N$ contains two types of elements, distinguished by the number of corridors in the corresponding path. For each element $[\eta_i,\eta_j]\in \mathcal{C}_{{ij}}^N$, the associated lower bound is given by:
\begin{enumerate}[label=\roman*.]
    \item Adjacent: The path consists of a single expansion corridor, specifically $\rho_{{ij}_{[i,j]}}:=\langle(i,j)\rangle$. Therefore, $\overline{\theta}_{ij}\geq w(\rho_{{ij}_{[i,j]}})=w_{ij}$. 
  \item Non-adjacent $[n_i,n_j]$: For neighboring pairs $[n_i,n_j]$, an SPP is solved to set $\overline{\theta}_{n_i n_j}={w(\underline{\rho}^0_{n_i n_j})}$. The tie-line weights $w_{in_i}$ and $w_{n_jj}$ are then added to obtain the full candidate path $\rho_{{ij}_{[n_i,n_j]}}:=\langle (i,n_i),\underline{\rho}^0_{n_i n_j},(n_j,j)\rangle$, which yields \vspace{-8pt}
  \begin{equation}
     \overline{\theta}_{ij} \geq w(\rho_{{ij}_{[n_i,n_j]}}) = w_{in_i}+w(\underline{\rho}^0_{n_i n_j})+w_{n_jj}.  \label{eq:caseiii}
  \end{equation}
\end{enumerate}
\begin{example}\label{ex:ex3a}
        Consider the disconnected bus pair $i_0$ and $i_1$ in Figure~\ref{fig:fig2}a, where the labels on the edges indicate line weights. Given that $(i_0,i_1) \in \Omega^1$, case (i) is applicable to the element $[i_0,i_1]\in\mathcal{C}^N_{i_0i_1}$, resulting in lower bound $\overline{\theta}_{i_0 i_1}\geq w(\rho_{{i_0i_1}_{[i_0,i_1]}})=w_{i_0i_1} = 1$.
        Additionally, for the neighbor pair $[i_4,i_5] \in \mathcal{C}^{N}_{i_0i_1}$, which corresponds to case (ii), the algorithm extends the shortest path between $i_4$ and $i_5$ to create a full path that connects $i_0$ to $i_1$, giving the lower bound \vspace{-7pt}
  \[\overline{\theta}_{i_0i_1} \geq w(\rho_{{i_0i_1}_{[i_4,i_5]}}) = w_{i_0i_4}+w(\underline{\rho}^0_{i_4i_5})+w_{i_5i_1}=1+1+1=3\,.\]
    \end{example}

When bus $i$ from a disconnected pair $[i,j]$ is connected to the initial grid (e.g., bus $i_0$ in Figure~\ref{fig:fig2}a and Figure~\ref{fig:fig2}b), it is necessary to factor the existing connections between bus $i$ and $G^0$ to ensure feasibility and applicability of Phase I. 
\vspace{3pt}
\begin{example}\label{eq:ex5}
    In Figure~\ref{fig:fig2}b, if the elements $[i_4,i_4]$, $[i_4,i_5]$, $[i_2,i_4]$, and $[i_2,i_5]$ are included in $\mathcal{C}^N_{{i_0i_6}}$, the shortest candidate path, namely $\rho_{{i_0i_6}_{(1)}}:=\langle(i_0,i_2),(i_2,i_4),(i_4,i_6)\rangle$, is excluded from consideration. Instead, the longer paths $\rho_{{i_0i_6}_{(2)}} :=\langle(i_0, i_4), (i_4, i_6)\rangle$, $\rho_{{i_0i_6}_{(3)}} :=\langle(i_0, i_4), (i_4, i_2), (i_2, i_5), (i_5, i_6)\rangle$, and $\rho_{{i_0i_6}_{(4)}} :=\langle(i_0, i_2), (i_2, i_5), (i_5, i_6)\rangle$ are obtained from case (ii) of Definition \ref{def:reach}. In effect, enumerating all ordered neighboring pairs $[n_{i_0},n_{i_6}]$ connected in $G$ to construct $\mathcal{C}^N_{i_0i_6}$ impedes obtaining the tightest value for $\overline{\theta}_{i_0i_6}$. Additionally, including $[i_1,i_4]$ in $\mathcal{C}^N_{i_0i_6}$ requires an established path between $i_1$ and $i_4$. However, since $i_4$ is already reachable from $i_0$, requiring connections from any $n_{i_0} \in N_G(i_0)$ to $i_4$ becomes unnecessary and restrictive. In summary, listing all pairs $[n_{i_0},n_{i_6}]$ may limit the applicability of Phase I based on the existing connectivity of $i_0$.
\end{example}
\begin{definition}\label{def:reach2}
    For a disconnected bus pair where one bus is connected to the current network, the reachability set is refined by replacing $[n_i, n_j]$ with $[i,n_j]$ if $(i,n_i)\in \Omega^0$ or $\mathcal{C}^{G^0}_{in_j}\neq \emptyset $, and with $[n_i,j]$ if $(n_j,j)\in \Omega^0$ or $\mathcal{C}^{G^0}_{n_ij}\neq \emptyset$.
\end{definition}
\noindent To simplify the upcoming lemma and proof, connections of the form $[i, n_j]$ and $[n_i, j]$ are jointly represented as $[i, n_j]$ since $G$ and $G^0$ are undirected (i.e., $[i, j]$ is equivalent to $[j, i]$).
\begin{lemma}\label{lemma:special}
     Refining the reachability set $\mathcal{C}^{N}_{ij}$ by replacing $[n_i, n_j]$ with $[i, n_j]$ when $(i, n_i) \in \Omega^0$ or $\mathcal{C}^{G^0}_{in_j} \neq \emptyset$ ensures that all necessary connections are included to capture the relevant paths between buses $i$ and $j$ through $n_j$.
  \end{lemma}
\begin{figure}
    \centering
    \includegraphics[width=0.8\linewidth]{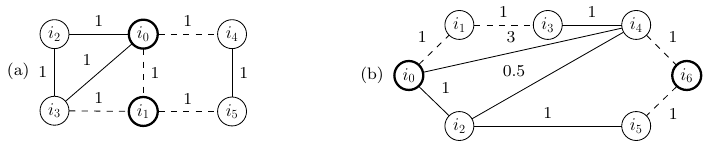}
    \caption{Examples of situations where the reachability set should be refined: (a) In $\mathcal{C}_{{i_0i_1}}^N$, $\{[i_2,i_3],[i_3,i_3]\}$ are replaced by $[i_0,i_3]$, as $(i_0,i_2), (i_0,i_3) \in \Omega^0$; (b) In $\mathcal{C}^N_{i_0i_6}$, $[i_4,i_4]$, $[i_4,i_5]$, and $[i_2,i_4]$ are replaced by $[i_0,i_4]$, and $[i_2,i_5]$ is replaced by $[i_0,i_5]$, since $(i_0,i_2),(i_0,i_4) \in \Omega^0$, and since $[i_0,i_4]$ is connected in $G^0$, $[i_0,i_4]$ also replaces $[i_1,i_4]$.}
    \label{fig:fig2}
\end{figure}
  \begin{proof}
       Pairs $[n_i, n_j]$ are replaced with $[i, n_j]$ in two cases:\\
\indent \textbf{Case 1:} $(i,n_i)\in \Omega^0$. Let $N^0(i)\subseteq N_G(i)$ denote the neighbors of bus $i$ with $(i,n_i)\in \Omega^0$. It is correct to replace $[n_i, n_j]$ with $[i, n_j]$ since the SPP accounts for traversing all neighbors $n_i \in N^0(i)$, when determining the shortest existing path from $i$ to $n_j$. When bus $i$ and at least two of its neighbors, say $n^1_i$ and $n^2_i$, form a cycle of established corridors $c:=\langle i, n^1_i, n^2_i, i \rangle$, there is a path with at least two corridors connecting each $n_i\in N^0(i)$ to $i$ by traversing through other neighbors in $N^0(i)$. In other words, this refinement ensures that the reachability set captures all paths that pass through multiple neighbors. As a result, all sub-paths in the refined reachability set are no more than one corridor away from the original buses.\\
\indent \textbf{Case 2:} $\mathcal{C}^{G^0}_{in_j}\neq \emptyset $ (i.e., bus $i$ has a path to $n_j$ in $G^0$). Assume that $n'_i\in N^0(i)$, and therefore the neighbor pair $[n'_i, n_j]\in \mathcal{C}^{N}_{ij}$ is replaced with $[i, n_j]$ due to case 1. We demonstrate that the candidate path through $[n_i,n_j]\in\mathcal{C}^{N}_{ij}$, where $(i, n_i),(n_j,j)\in \Omega^1$, should be excluded from $\mathcal{C}^{N}_{ij}$. Constructing the full path through $[n_i, n_j]$ requires adding lines across both corridors $(i, n_i)$ and $(n_j,j)$, whereas constructing the path corresponding to $[i,n_j]$ only necessitates building the corridor $(n_j, j)$. More specifically, the construction of both lines is redundant and restrictive in this case: when the path for $[n_i, n_j]$ is constructed, the path for $[i, n_j]$ is automatically formed, whereas connections where only $(n_j, j)$ is constructed are incorrectly discarded (since a path between $i$ and $j$ is still formed).
  \end{proof}
 This refinement leads to a third type of elements in $\mathcal{C}_{ij}^N$: 
\begin{enumerate}[label=\roman*.]
    \setcounter{enumi}{2} 
    \item Non-adjacent $[i,n_j]$ (when $N^0(i)\neq \emptyset$): For connections of the form $[i,n_j]$, $\overline{\theta}_{in_j}$ is obtained as $w(\underline{\rho}^0_{in_j})$ by solving the SPP for $[i,n_j]$. Then, the terminal corridor weight $w_{n_jj}$ is added to create the candidate path $\rho_{{ij}_{[i,n_j]}}$, providing the corresponding lower bound:\vspace{-6pt}
    \begin{equation}
        \overline{\theta}_{ij}\geq  w(\rho_{{ij}_{[i,n_j]}}) = w(\underline{\rho}^0_{in_j})+w_{n_jj}. \label{eq:caseii}
    \end{equation}
    \end{enumerate}
  \begin{example}
      In Example \ref{ex:ex3a}, the reachability set $\mathcal{C}^{N}_{i_0i_1}=\{[i_0,i_1],[i_2,i_3],[i_3,i_3],[i_4,i_5]\}$ is refined to $\mathcal{C}^{N}_{i_0i_1}=\{[i_0,i_1],[i_0,i_3],[i_4,i_5]\}$ by replacing $[i_2,i_3]$ and $[i_3,i_3]$ with $[i_0,i_3]$, as $(i_0,i_2),(i_0,i_3)\in \Omega^0$. The lower bound for $\overline{\theta}_{i_0i_1}$ corresponding to the element $[i_0, i_3]$ is established as \vspace{-7pt}
  \[\overline{\theta}_{i_0i_1}\geq w(\rho_{{i_0i_1}_{[i_0,i_3]}}) =  w(\underline{\rho}^0_{i_0i_3})+w_{i_3i_1}=1+1=2\,.\]
    In Example \ref{eq:ex5}, the set $\mathcal{C}^{N}_{i_0i_6}=\{[i_1,i_4],[i_1,i_5],[i_4,i_4],[i_4,i_5],[i_2,i_4],[i_2,i_5]\}$ is similarly refined to $\{[i_0,i_4],[i_0,i_5]\}$. This yields lower bounds of 2.5 and 3 for $\overline{\theta}_{i_0i_6}$, and thus $\overline{\theta}_{i_0i_6}=w(\hat{\rho}_{i_0i_6})=3$, which is tighter than the LPP-based value $w(\overline{\rho}_{i_0i_6})=5.5$.
  \end{example}
\begin{algorithm}
    \caption{Longest Shortest-Path Connection Algorithm (LSPC) - Phase I}\label{alg:LSPC1}
  \begin{algorithmic}[1]
        \State {$\textbf{Inputs: } G,G^0,[i,j],\gamma$}
        \State {\textbf{Output:} A tighter upper bound on the angle difference between buses $i$ and $j$}
               \State $\overline{\theta}_{ij} \gets 0$            \Comment{Initialize $\overline{\theta}_{ij}$}
               \State \textbf{Construct} $\mathcal{C}^N_{ij}$ \Comment{Construct the reachability set by applying Definition \ref{def:reach}}
               \If {$\exists [n_i,n_j]\in \mathcal{C}^N_{ij}\mid (i,n_i)\in \Omega^0 \textbf{ or } \mathcal{C}^{G^0}_{in_j}\neq \emptyset$}
                   \State \textbf{Refine} $\mathcal{C}^N_{ij}$  \Comment{Refine the reachability set by applying Definition \ref{def:reach2}}
               \EndIf
                \If {$\mathcal{C}^{G^0}_{[\eta_i,\eta_j]} \neq \emptyset,\forall\ [\eta_i,\eta_j] \in \mathcal{C}^N_{ij}\backslash \{[i,j]\}$}  \Comment{Verify Phase I's applicability to $[i, j]$}
                    \ForAll {$[\eta_i,\eta_j]\in\mathcal{C}^N_{ij}$}   \Comment{Compute $w(\rho_{{ij}_{[\eta_i,\eta_j]}})$ using cases (i),(ii), and (iii)}
                        \State $\overline{\theta}_{ij} \gets \max\{\overline{\theta}_{ij},w(\rho_{{ij}_{[\eta_i,\eta_j]}})\}$
                    \EndFor
                \Else
                    \State $\overline{\theta}_{ij} \gets \gamma$            \Comment{Phase I fails}
                \EndIf
    \end{algorithmic}
\end{algorithm}
 The following theorem formally establishes the tightest angle-difference bound for an initially disconnected bus pair by computing the weights of candidate paths associated with the three types of elements of the reachability set.
\begin{theorem} \label{lemma:lemma_LSPC}
For a disconnected bus pair \([i,j]\), if every pair in \(\mathcal{C}_{ij}^N\setminus\{[i,j]\}\) is connected in \(G^0\), then the tightest path-based angle-difference bound is obtained by the longest shortest-path connection in the refined set \(\mathcal{C}_{ij}^N\) (considering cases i, ii, and iii), namely
\begin{equation}\label{eq:LSPC1result}
    |\theta_{i}-\theta_{j}| \leq \underset{{[\eta_i,\eta_j]}\in \mathcal{C}^N_{{ij}}}{\max} \{w({\rho_{{ij}_{[\eta_i,\eta_j]}}})\}=:w(\hat{\rho}_{ij}),
\end{equation}
where the weight of each shortest-path connection $\rho_{{ij}_{[\eta_i,\eta_j]}}$ is computed as: \vspace{-7pt}
\begin{subequations}\label{eq:reachability}
\begin{flalign}
  &   w(\rho_{{ij}_{[\eta_i,\eta_j]}})=  w_{ij}, && \forall \; [\eta_i,\eta_j] = [i,j] \\
   & w(\rho_{{ij}_{[\eta_i,\eta_j]}})=  w(\underline{\rho}^0_{in_j}) +w_{n_jj},  && \forall \; [\eta_i,\eta_j]\in\{[i,n_j]\} \\
   &   w(\rho_{{ij}_{[\eta_i,\eta_j]}})= w_{in_i} + w(\underline{\rho}^0_{n_in_j}) + w_{n_jj}, && \forall \;[\eta_i,\eta_j]\in\{[n_i,n_j]\}. \label{eq:opercount0}
\end{flalign}
\end{subequations}
\end{theorem}
\begin{proof}
 To ensure the reachability set $\mathcal{C}^N_{{ij}}$ captures all relevant paths between buses \(i\) and \(j\), all ordered neighboring pairs $[n_i,n_j]$ connected in $G$, as well as the pair $[i,j]$ when $(i,j)\in \Omega^1$, are enumerated by applying \eqref{eq:reachabilityset}.
 To extend the applicability of the LSPC algorithm and derive the tightest upper bound, $\mathcal{C}^N_{ij}$ is refined following Definition \ref{def:reach2}, as established by Lemma \ref{lemma:special}. Assuming all $[\eta_i, \eta_j] \in \mathcal{C}^N_{ij} \setminus \{[i, j]\}$ are connected in $G^0$, the weight of the candidate path corresponding to $[\eta_i, \eta_j] \in \mathcal{C}^N_{ij}$ is calculated as follows:
\newline\indent \textbf{Case i:} For $[i, j] \in \mathcal{C}^N_{ij}$, if a line is established along the expansion corridor $(i,j)$, then inequality \eqref{eq:capacity-reactance} yields\vspace{-5pt}
\begin{equation}
    \left|\theta_i-\theta_j\right| \leq  w(\rho_{{ij}_{[i,j]}}) = w_{ij}. \label{eq:up1}
\end{equation}
\indent \textbf{Case ii:} For $[n_i,n_j]\in \mathcal{C}^N_{{ij}}$, we employ the SPP to determine their angle-difference bounds $\overline{\theta}_{n_i n_j}$. By extending the terminal expansion corridors (with weights $w_{i n_i}$ and $w_{n_j j}$) on both sides of each SPP-based sub-path, a candidate path between $i$ and $j$ is formed. Assuming the construction of both $(i,n_i)$ and $(n_j,j)$, the path is added to the network, yielding the angle-difference inequality
\begin{subequations}
    \begin{flalign} 
  \left|\theta_i-\theta_j\right| \leq &\quad  \underbrace{|\theta_{i} - \theta_{n_i}|}_{\substack{\leq w_{in_i}}} +  \underbrace{|\theta_{n_i} - \theta_{n_j}|}_{\substack{\leq w(\underline{\rho}^0_{n_i n_j})}} +  \underbrace{|\theta_{n_j} - \theta_{j}|}_{\substack{\leq w_{n_jj}}} \label{eq:up21} \\[-2pt]
 \Rightarrow |\theta_{i}-\theta_{j}| \leq & \quad w_{in_i}+ w(\underline{\rho}^0_{n_i n_j}) +w_{n_jj}=w(\rho_{{ij}_{[n_i,n_j]}}).  \label{eq:up22}
  \end{flalign}
\end{subequations}
\indent \textbf{Case iii:} For the candidate paths represented by $[i, n_j]$, constructing the corridor $(n_j, j)$ extends the SPP-based sub-path from $i$ to $n_j$ and is required to guarantee that
\begin{subequations}
    \begin{flalign} 
  \left|\theta_i-\theta_j\right| \leq &\quad   \underbrace{|\theta_{i} - \theta_{n_j}|}_{\substack{\leq w(\underline{\rho}^0_{in_j}})} +  \underbrace{|\theta_{n_j} - \theta_{j}|}_{\substack{\leq w_{n_jj}}}\label{eq:up31} \\[-3pt]
 \Rightarrow |\theta_{i}-\theta_{j}| \leq & \quad  w(\underline{\rho}^0_{in_j}) +w_{n_jj}=w(\rho_{{ij}_{[i,n_j]}}).\label{eq:up32} 
  \end{flalign}
\end{subequations}
Since inequalities \eqref{eq:up1},\eqref{eq:up22}, and \eqref{eq:up32} are derived from prospective paths not yet included in the network, it is essential to ensure that $\overline{\theta}_{ij} \geq \underset{[\eta_i, \eta_j] \in \mathcal{C}^N_{ij}}{\max} w(\rho_{{ij}_{[\eta_i, \eta_j]}})$. To maintain the validity and tightness of $\overline{\theta}_{ij}$, it must be at least as large as the maximum lower bound, thereby establishing inequality \eqref{eq:LSPC1result}.

To show that inequality \eqref{eq:LSPC1result} provides the tightest angle-difference bound for $[i,j]$, note that any smaller bound would require a shorter path for some pair in $\mathcal{C}^N_{ij}$. This cannot occur since, by Lemma~\ref{lemma:special}, the reachability set retains only the necessary connection cases, each intermediate sub-path is obtained by solving an SPP, and each full path appends at most one candidate line at each end, whose weights cannot be further tightened. Hence, the longest shortest-path connection yields the tightest attainable path-based bound, and LSPC is exact under the stated connectivity condition when no additional information about the construction of expansion corridors is available.
\end{proof}
Having shown that LSPC yields tight angle-difference bounds, we next prove that the resulting path-based VIs dominate their LPP-based counterparts.
\begin{proposition}\label{prop:prop2}
Let \(\rho_{ij}\) be a candidate path connecting buses \(i\) and \(j\) in \(G\). Since the path-based VI associated with \(\rho_{ij}\) cannot enforce an angle-difference bound smaller than \(w(\rho_{ij})\), the relevant comparison region is \( |\theta_i-\theta_j| \geq w(\rho_{ij}) \). Over this region, the path-based VI constructed using \(w(\hat{\rho}_{ij})\) dominates those derived with LPP or with any initial bound greater than \(w(\hat{\rho}_{ij})\).
\end{proposition}
\begin{proof}
    See Appendix \ref{secA2}.
\end{proof}
We now complement the results by analyzing the computational complexity of deriving the LSPC-based bounds.\smallskip
\begin{proposition}
Let \(n := |\mathcal{B}|\), \(m := |\Omega^0\cup\Omega^1|\), and assume that \(P(\deg(i)>K)\sim \exp(-0.5K)\) for all \(i\in\mathcal{B}\) and \(K\in\mathbb{N}\). The time complexity of deriving \(\overline{\theta}_{ij}\) using LSPC Phase I is \(\mathcal{O}(m\log n)\).
\end{proposition}
\begin{proof}
In the worst case, \((i,j)\in\Omega^1\), and the reachability set \(\mathcal{C}^N_{ij}\) includes every neighboring pair \([n_i,n_j]\in N_G(i)\times N_G(j)\), all of which are connected in \(G^0\), so that Phase~I is applicable. LSPC solves an SPP in \(G^0\) from each $n_i\in N_G(i)$ when $|N_G(i)|\leq|N_G(j)|$, and from each $n_j\in N_G(j)$ otherwise. Using a heap-based implementation of Dijkstra's algorithm \cite{cormen2022introduction}, this takes \(\mathcal{O}(\min\{\deg(i),\deg(j)\}\, m\log n)\) time. For each \([n_i,n_j]\in N_G(i)\times N_G(j)\), the algorithm then constructs a shortest-path connection $\rho_{ij_{[n_i,n_j]}}$ by appending the terminal lines \((i,n_i)\) and \((n_j,j)\) to \(\underline{\rho}^0_{n_i n_j}\), which requires constant time per pair. Over all such pairs, this contributes \(\mathcal{O}(\deg(i)\deg(j))\). The additional check for \((i,j)\in\Omega^1\) contributes only constant time. Therefore, the overall time complexity of LSPC Phase~I is $\mathcal{O}\!\left(\min\{\deg(i),\deg(j)\}\, m\log n + \deg(i)\deg(j)\right)$. Under the empirically observed exponential decay of the bus-degree distribution in practical power grids \cite{albert2004structural}, bus degrees are effectively \(\mathcal{O}(1)\), yielding an overall complexity of \(\mathcal{O}(m\log n)\).
\end{proof}
\begin{figure}
        \centering
        \includegraphics[width=1\linewidth]{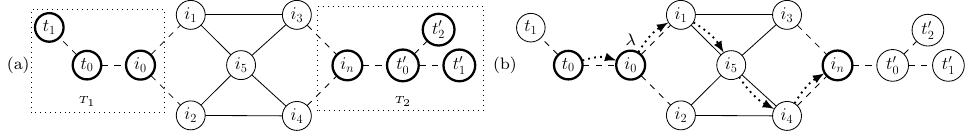}
       \caption{Integration of expansion trees into the grid: (a) buses in expansion trees $T_1$ and $T_2$ before integration; (b) Phase II combines the Phase I-based sub-paths $\hat{\rho}_{t_0i_0}:=\langle(t_0,i_0)\rangle$ and $\hat{\rho}_{i_0i_n}:=\langle(i_0,i_1),(i_1,i_5),(i_5,i_4),(i_4,i_n)\rangle$ to link $t_0$ to $i_n$ via the intermediary bus $i_0$.}
        \label{fig:fig3}
    \end{figure}
\subsection{LSPC Phase II: Tree-structured bus integration}\label{subsec42}
LSPC Phase II extends Phase I to disconnected bus pairs involving new buses located beyond the root of an expansion tree. The key observation is that any path to such a bus must pass through the tree root—whose connection to the existing grid is obtained through Phase I—and then through the intermediary buses along the unique tree path to the target bus. Formally, Phase II applies to disconnected buses $i$ and $j$ when there exists an \emph{intermediary bus} $\lambda$ such that relevant paths from $i$ to $\lambda$ and from $\lambda$ to $j$, with corresponding bounds $\overline{\theta}_{i\lambda}$ and $\overline{\theta}_{\lambda j}$, have already been identified.
\vspace{3pt}
\begin{example}
    In Figure~\ref{fig:fig3}a, $T_1$ and $T_2$ exemplify expansion trees, where the buses in these trees illustrate the applicability of Phase II. In this network, when power flows into bus $i_n$ from $i_3$ or $i_4$, the only pathway for the power to reach bus $t'_1$ is through bus $t'_0$. As a result, once a path to $i_n$ is built, it can be easily extended along the tree structure rooted at $i_n$ to reach additional new buses located deeper in the tree.\\
    Consider now the buses $t_0$ and $t'_0$. Once power flows into bus $i_0$ from $t_0$, the sub-network between $i_0$ and $i_n$ can be bypassed through the Phase I-based path $\hat{\rho}_{i_0i_n}$ (the dotted arrows originating from $i_0$) to reach $t'_0$. Specifically, by merging this sub-path with two Phase I-based sub-paths, $\hat{\rho}_{t_0i_0}:=\langle(t_0,i_0)\rangle$ and $\hat{\rho}_{i_nt'_0}:=\langle(i_n,t'_0)\rangle$, a relevant path connecting $t_0$ and $t'_0$, denoted as $\hat{\rho}_{t_0t'_0}:=\langle(t_0,i_0),\hat{\rho}_{i_0i_n},(i_n,t'_0)\rangle$, is obtained.
\end{example}
\vspace{3pt}
 From the set of all buses reachable from both $i$ and $j$, the intermediary bus $\lambda$ that provides the shortest complete path, formed by combining $\rho_{i\lambda}$ and $\rho_{\lambda j}$, should be selected to avoid \textit{non-simple} paths, i.e., traversing a bus more than once, and thereby prevent loose upper bounds. This implies that Phase II can be applied iteratively, utilizing previously established connections to construct paths to new buses located deeper into expansion trees.
\begin{example}
    Consider the bus pair $[t_0,i_n]$ in Figure~\ref{fig:fig3}b. Although Phase I is not directly applicable---as no established path connects $i_0\in N_G(t_0)$ to buses in $N_G(i_n)$---it can be applied to another bus in the tree, namely $i_0$, to establish bounds $\overline{\theta}_{t_0i_0}=1$ and $\overline{\theta}_{i_0i_n}=4$. Using $\lambda=i_0$ as the intermediary bus and merging these paths, the angle-difference bound associated with the resulting complete path is determined as:\vspace{-4pt}
\begin{equation}
  |\theta_{t_0}-\theta_{i_n}|\leq \underbrace{|\theta_{t_0}-\theta_{i_0}|}_{\substack{\leq \overline{\theta}_{t_0i_0}=1}} + \underbrace{|\theta_{i_0}-\theta_{i_n}|}_{\substack{\leq \overline{\theta}_{i_0i_n} = 4}} \leq 5\,. \label{eq:LSPC21}
\end{equation}
The path obtained between $t_0$ and $i_n$ can be merged with the path $\hat{\rho}_{i_nt'_0}$ in the subsequent iteration of Phase II, thereby building a connection between $t_0$ and $t'_0$.
\end{example}
\begin{algorithm}[H]
    \caption{Longest Shortest-Path Connection Algorithm (LSPC) - Phase II}\label{alg:LSPC2}
    \begin{algorithmic}[1]
    \State {$\textbf{Inputs: } i,j\in \mathcal{B},\gamma,\overline{\theta}_{ij}$} \Comment{SPP or Phase I bounds (or the $\gamma$ bound, if they fail)}
    \State {\textbf{Outputs: } A tighter bound on the angle difference between buses $i$ and $j$}
            \ForAll {$ \lambda \in \mathcal{B}\backslash\{i,j\}$}
                \If {$ \overline{\theta}_{i\lambda} < \gamma$ \textbf{ and } $\overline{\theta}_{\lambda j} < \gamma$}
                    \State $\overline{\theta}_{ij} \gets \min\{\overline{\theta}_{ij},\overline{\theta}_{i\lambda} + \overline{\theta}_{\lambda j}\}$
                \EndIf
            \EndFor
\end{algorithmic}
\end{algorithm}
\begin{proposition}\label{lemma:lemma_LSPC2}
Consider buses $i$ and $j$ that are disconnected in the initial network, with at least one of them belonging to an expansion tree. If bounds $\overline{\theta}_{i\lambda}$ and $\overline{\theta}_{\lambda j}$ can be determined for at least one $\lambda \in \mathcal{B} \backslash \{i, j\}$, then $\overline{\theta}_{ij}$ can be obtained as:\vspace{-4pt}
\begin{equation}
|\theta_{i}-\theta_{j}| \leq \min_{\substack{\lambda \in \mathcal{B}\backslash\{i, j\}, (\overline{\theta}_{i \lambda}<\gamma \;\land \;\overline{\theta}_{\lambda j}<\gamma)}} \{\overline{\theta}_{i\lambda}  +  \overline{\theta}_{\lambda j}\}=:w(\hat{\rho}_{ij}). \label{eq:LSPC2main}
\end{equation}
\end{proposition}
\begin{proof}
With the angle-difference bounds $\overline{\theta}_{i\lambda}$ and $\overline{\theta}_{\lambda j}$ from Phase I, a complete path between $i$ and $j$ can be formed by linking the paths connecting them to $\lambda$. Traversing this full path creates a telescoping effect on the left side, with the sum of the weights of the sub-paths, $\overline{\theta}_{i\lambda}+\overline{\theta}_{\lambda j}$, yielding an upper bound on the overall angle difference:\vspace{-4pt}
\begin{equation}
    |\theta_{i}-\theta_{j}|\leq \underbrace{|\theta_{i}-\theta_{\lambda}|}_{\substack{\leq \overline{\theta}_{i\lambda}}} + \underbrace{|\theta_{\lambda}-\theta_{j}|}_{\substack{\leq \overline{\theta}_{\lambda j}}} \leq \overline{\theta}_{i\lambda}+\overline{\theta}_{\lambda j}.
\end{equation}
The presence of multiple intermediary buses $\lambda$ results in multiple upper bounds. The smallest is selected to yield the strongest valid upper bound from these options, yielding inequality \eqref{eq:LSPC2main}. 
\end{proof}
\section{Conclusion}\label{sec5}
This paper introduces the longest shortest-path connection (LSPC) algorithm to derive tight upper bounds on voltage-angle differences between disconnected bus pairs when expansion involves integrating relatively simple new-bus subnetworks. LSPC is a polynomial-time algorithm that overcomes the practical limitations of the longest path problem (LPP) approach, which the existing literature generally assumes is necessary to obtain provably tight bounds for disconnected bus pairs. This paper also shows that existing path-based valid inequalities for DC transmission expansion planning (DC-TEP), when initialized with LSPC, dominate those derived from the LPP or any other larger initial bound. In future work, we will focus on evaluating the computational benefits of the proposed inequalities in solving large-scale DC-TEP problems.
\section*{Data availability}
This article does not contain any datasets.
\begin{appendices}
\section{}\label{secA1}
We define $\rho_{ij}$ as a candidate path in $G$ connecting buses $i$ and $j$ that includes at least one expansion corridor. Let $N_e(\rho_{ij})$ be the count of such expansion corridors along $\rho_{ij}$. The complete form of the path-based VIs (see \eqref{eq:simplevi}) for candidate paths $\rho_{ij}$ is provided in the following proposition. We also demonstrate that when $w(\rho_{ij})$ exceeds the initial angle-difference bound $\overline{\theta}_{ij}$, it may exclude integer-feasible solutions.
    \begin{proposition} \label{theorem:theo1}The following expression, which provides the complete form of the path-based inequalities, is valid only for candidate paths $\rho_{ij} \in \mathcal{C}^{G}_{ij}$ with $w(\rho_{ij}) < \overline{\theta}_{ij}$:\vspace{-3pt}
        \begin{equation}
        \left|\theta_i-\theta_j\right| \leq w\left(\rho_{ij}\right)+\left(\overline{\theta}_{ij}-w\left(\rho_{ij}\right)\right)\Big(N_e\left(\rho_{ij}\right)-\sum_{(k, l) \in \rho_{ij}} \mathbb{I}_{kl} y_{kl}\Big). \label{eq:dyvi}
        \end{equation}
        Here, the indicator function $\mathbb{I}_{kl}$ takes a value of 1 if $(k,l) \in \rho_{ij} \cap\Omega^1 $ and 0 otherwise.
        \end{proposition}
        \begin{proof}
        We will show that absent the condition $w(\rho_{ij})<\overline{\theta}_{ij}$, inequalities \eqref{eq:dyvi} may eliminate integer-feasible solutions. For succinctness, define the variable $\nu:=\left(N_e\left(\rho_{ij}\right)-\sum_{(k,l) \in \rho_{ij}} \mathbb{I}_{kl} y_{kl}\right)$ and use it to reformulate inequality \eqref{eq:dyvi} as\vspace{-3pt}
        \begin{equation}
            |\theta_i-\theta_j|\leq(1-\nu)\; w(\rho_{ij}) + \nu \; \overline{\theta}_{ij}.\label{eq:dyviref}
        \end{equation}
        When $w(\rho_{ij})>\overline{\theta}_{ij}$ and $\nu>1$ (i.e., more than one unbuilt expansion corridor exists along $\rho_{ij}$), the right-hand side of \eqref{eq:dyviref} becomes smaller than the initial bound $\overline{\theta}_{ij}$, despite no path shorter than $\overline{\theta}_{ij}$ having been constructed. Consequently, inequality \eqref{eq:dyviref} becomes invalid.
    \end{proof}
 \section{}\label{secA2}
 Let $w(\hat{\rho}_{ij})$ represent the angle-difference bound obtained using the LSPC algorithm for buses $i$ and $j$. Proposition \ref{prop:prop2} evaluates the LSPC-based bound against the LPP-based bound to compare the tightness of the resulting path-based VIs associated with their respective bounds.
\begin{restatedproposition}{prop:prop2}
Let \(\rho_{ij}\) be a candidate path connecting buses \(i\) and \(j\) in \(G\). Since the path-based VI associated with \(\rho_{ij}\) cannot enforce an angle-difference bound smaller than \(w(\rho_{ij})\), the relevant comparison region is \( |\theta_i-\theta_j| \geq w(\rho_{ij}) \). Over this region, the path-based VI constructed using \(w(\hat{\rho}_{ij})\) dominates those derived with LPP or with any initial bound greater than \(w(\hat{\rho}_{ij})\).
\end{restatedproposition}
\begin{proof}
Consider the path-based VIs derived using $w(\hat{\rho}_{ij})$ and $w(\overline{\rho}_{ij})$ as the initial bound $\overline{\theta}_{ij}$, respectively:\vspace{-6pt}
\begin{equation}
     \sum_{(k, l) \in \rho_{ij}}^{N_e} \mathbb{I}_{kl} y_{kl} \leq N_e(\rho_{ij}) + \frac{w(\rho_{ij})-|\theta_i-\theta_j|}{w(\hat{\rho}_{ij})-w(\rho_{ij})},\label{eq:poly1}
\end{equation}\vspace{-6pt}
\begin{equation}
    \sum_{(k,l) \in \rho_{ij}}^{N_e} \mathbb{I}_{kl} y_{kl} \leq N_e(\rho_{ij}) + \frac{w(\rho_{ij})-|\theta_i-\theta_j|}{w(\overline{\rho}_{ij})-w(\rho_{ij})}.\label{eq:poly2}
\end{equation}
These inequalities can be represented as $\pi x\leq \pi^1_0$ and $\pi x\leq \pi^2_0$. We need to demonstrate that $\pi^1_0$, i.e., the right-hand side of \eqref{eq:poly1}, is smaller than $\pi^2_0$, i.e., the right-hand side of \eqref{eq:poly2} within the effective domain, namely, whenever $|\theta_i-\theta_j|\geq w(\rho_{ij})$.
From the definition of the longest path, we have that $w(\hat{\rho}_{ij})\leq w(\overline{\rho}_{ij})$, with equality indicating the worst-case scenario in which all simple paths between $i$ and $j$ have identical weights. Excluding the case of equality, two distinct cases arise for the resulting path-based VIs to be comparable:\\
\indent \textbf{Case 1:} $w(\rho_{ij})=w(\hat{\rho}_{ij})<w(\overline{\rho}_{ij})$. The path $\rho_{ij}$ offers no further improvement to the LSPC-based bound, and a path-based VI can only be derived setting $\overline{\theta}_{ij}=w(\overline{\rho}_{ij})$. The VI eliminates solutions from the relaxed problem's space when $|\theta_i - \theta_j| \geq w(\rho_{ij})$, but this removal is redundant, as the VI $|\theta_i - \theta_j| \leq w(\hat{\rho}_{ij}) = w(\rho_{ij})$ dominates it. Furthermore, the VI becomes ineffective when $|\theta_i - \theta_j| \leq w(\rho_{ij})$, as it is dominated by the trivial \textit{corridor enumeration} VI, formulated as $\sum_{(k,l) \in \rho_{ij}}^{N_e} \mathbb{I}_{kl} y_{kl} \leq N_e(\rho_{ij})$. This is because
\begin{equation}
    N_e(\rho_{ij}) \leq N_e(\rho_{ij}) + \frac{w(\rho_{ij})-|\theta_i - \theta_j|}{w(\overline{\rho}_{ij}) - w(\rho_{ij})} \label{eq:propproof1}
\end{equation}
(the numerator is guaranteed to be non-negative, and the denominator is positive in this case). Therefore, \eqref{eq:poly2} does not provide a tighter VI than the LSPC-based and the corridor enumeration VI.
\newline\indent\textbf{Case 2:} $w(\rho_{ij})< w(\hat{\rho}_{ij}) < w(\overline{\rho}_{ij})$. Both initial bounds $w(\hat{\rho}_{ij})$ and $w(\overline{\rho}_{ij})$, can be used to derive the path-based VIs \eqref{eq:poly1} and \eqref{eq:poly2}, respectively. When $|\theta_i-\theta_j|\leq w(\rho_{ij})$, both VIs are dominated by $\sum_{(k,l) \in \rho_{ij}} \mathbb{I}_{kl} y_{kl} \leq N_e(\rho_{ij})$, that is,
\begin{equation}
    N_e(\rho_{ij}) \leq N_e(\rho_{ij}) + \frac{w(\rho_{ij})-|\theta_i-\theta_j|}{w(\overline{\rho}_{ij})-w(\rho_{ij})} \leq N_e(\rho_{ij}) + \frac{w(\rho_{ij})-|\theta_i-\theta_j|}{w(\hat{\rho}_{ij})-w(\rho_{ij})},\label{eq:propproof2}
\end{equation}
since the fractions have an identical numerator, and $w(\overline{\rho}_{ij})-w(\rho_{ij})>w(\hat{\rho}_{ij})-w(\rho_{ij})>0$. Conversely, when $|\theta_i-\theta_j|\geq w(\rho_{ij})$, the LSPC-based VI dominates the other two, as a larger fraction term is subtracted from $N_e(\rho_{ij})$:
\begin{equation}
    N_e(\rho_{ij})+ \frac{w(\rho_{ij})-|\theta_i-\theta_j|}{w(\hat{\rho}_{ij})-w(\rho_{ij})} \leq N_e(\rho_{ij})+ \frac{w(\rho_{ij})-|\theta_i-\theta_j|}{w(\overline{\rho}_{ij})-w(\rho_{ij})}\leq N_e(\rho_{ij}).
\end{equation}
This demonstrates that when $|\theta_i-\theta_j|\geq w(\rho_{ij})$, the LSPC-based VI provides a tighter bound than that derived using $w(\overline{\rho}_{ij})$ or any path longer than $w(\hat{\rho}_{ij})$.

Consequently, within the effective domain, the LPP-based VIs are either dominated by LSPC-based angle-difference VIs (case 1) or the path-based VIs with LSPC-based initial bounds (case 2).
\end{proof}
\end{appendices}
\bibliographystyle{unsrtnat}
\bibliography{references}

@article{lumbreras2016new,
  title={The new challenges to transmission expansion planning. Survey of recent practice and literature review},
  author={Lumbreras, Sara and Ramos, Andr{\'e}s},
  journal={Electric Power Systems Research},
  volume={134},
  pages={19--29},
  year={2016},
  doi = {10.1016/j.epsr.2015.10.013},
  publisher={Elsevier}
}

@article{kirby2007evaluating,
  title={Evaluating transmission costs and wind benefits in Texas: Examining the ERCOT CREZ transmission study},
  author={Kirby, Brendan},
  journal={The Wind Coalition and Electric Transmission Texas, LLC, Texas PUC Docket},
  number={33672},
  year={2007}
}

@incollection{du2023renewable,
  title={Renewable Integration at ERCOT},
  author={Du, Pengwei},
  booktitle={Renewable Energy Integration for Bulk Power Systems: ERCOT and the Texas Interconnection},
  pages={1--26},
  address={Switzerland},
  year={2023},
  publisher={Springer}
}

@article{kocuk2016cycle,
  title={A cycle-based formulation and valid inequalities for DC power transmission problems with switching},
  author={Kocuk, Burak and Jeon, Hyemin and Dey, Santanu S and Linderoth, Jeff and Luedtke, James and Sun, Xu Andy},
  journal={Operations Research},
  volume={64},
  number={4},
  pages={922--938},
  year={2016},
  doi = {10.1287/opre.2015.1471},
  publisher={INFORMS}
}

@article{binato2001new,
  title={A new Benders decomposition approach to solve power transmission network design problems},
  author={Binato, Silvio and Pereira, M{\'a}rio Veiga F and Granville, S{\'e}rgio},
  journal={IEEE Transactions on Power Systems},
  volume={16},
  number={2},
  pages={235--240},
  year={2001},
  doi = {10.1109/59.918292},
  publisher={IEEE}
}

@article{oertel2014complexity,
  title={Complexity of transmission network expansion planning: NP-hardness of connected networks and MINLP evaluation},
  author={Oertel, David and Ravi, R},
  journal={Energy systems},
  volume={5},
  number={1},
  pages={179--207},
  year={2014},
  doi={doi:10.1007/s12667-013-0091-3},
  publisher={Springer}
}

@book{schrijver2003combinatorial,
  title={Combinatorial optimization: polyhedra and efficiency},
  author={Schrijver, Alexander and others},
  volume={24},
  number={2},
  address   = {Berlin},
  year={2003},
  publisher={Springer}
}

@article{garver1970transmission,
  title={Transmission network estimation using linear programming},
  author={Garver, Len L},
  journal={IEEE Transactions on power apparatus and systems},
  number={7},
  pages={1688--1697},
  year={1970},
  doi = {10.1109/TPAS.1970.292825},
  publisher={IEEE}
}

@article{dey2022node,
  title={Node-based valid inequalities for the optimal transmission switching problem},
  author={Dey, Santanu S and Kocuk, Burak and Redder, Nicole},
  journal={Discrete Optimization},
  volume={43},
  pages={100683},
  year={2022},
  doi = {10.1016/j.disopt.2021.100683},
  publisher={Elsevier}
}

@article{skolfield2022derivation,
  title={Derivation and generation of path-based valid inequalities for transmission expansion planning},
  author={Skolfield, J Kyle and Escobar, Laura M and Escobedo, Adolfo R},
  journal={Annals of Operations Research},
  volume={312},
  number={2},
  pages={1031--1049},
  year={2022},
  doi = {10.1007/s10479-022-04643-1},
  publisher={Springer}
}

@article{pan2020deepopf,
  title={Deepopf: A deep neural network approach for security-constrained dc optimal power flow},
  author={Pan, Xiang and Zhao, Tianyu and Chen, Minghua and Zhang, Shengyu},
  journal={IEEE Transactions on Power Systems},
  volume={36},
  number={3},
  pages={1725--1735},
  year={2020},
  publisher={IEEE}
}

@inproceedings{minot2016parallel,
  title={A parallel primal-dual interior-point method for DC optimal power flow},
  author={Minot, Ariana and Lu, Yue M and Li, Na},
  booktitle={2016 Power Systems Computation Conference (PSCC)},
  pages={1--7},
  year={2016},
  doi = {10.1109/PSCC.2016.7540826},
  organization={IEEE}
}

@article{kargarian2016toward,
  title={Toward distributed/decentralized DC optimal power flow implementation in future electric power systems},
  author={Kargarian, Amin and Mohammadi, Javad and Guo, Junyao and Chakrabarti, Sambuddha and Barati, Masoud and Hug, Gabriela and Kar, Soummya and Baldick, Ross},
  journal={IEEE Transactions on Smart Grid},
  volume={9},
  number={4},
  pages={2574--2594},
  year={2016},
  doi = {10.1109/TSG.2016.2614904},
  publisher={IEEE}
}

@article{horsch2018linear,
  title={Linear optimal power flow using cycle flows},
  author={H{\"o}rsch, Jonas and Ronellenfitsch, Henrik and Witthaut, Dirk and Brown, Tom},
  journal={Electric Power Systems Research},
  volume={158},
  pages={126--135},
  year={2018},
  doi = {10.1016/j.epsr.2017.12.034},
  publisher={Elsevier}
}

@article{abdi2021metaheuristics,
  title={Metaheuristics and transmission expansion planning: A comparative case study},
  author={Abdi, Hamdi and Moradi, Mansour and Lumbreras, Sara},
  journal={Energies},
  volume={14},
  number={12},
  pages={3618},
  year={2021},
  doi = {10.3390/en14123618},
  publisher={MDPI}
}

@inproceedings{mohammadi2013benders,
  title={A benders decomposition approach to corrective security constrained OPF with power flow control devices},
  author={Mohammadi, Javad and Hug, Gabriela and Kar, Soummya},
  booktitle={2013 IEEE Power \& Energy Society General Meeting},
  pages={1--5},
  year={2013},
  doi = {10.1109/PESMG.2013.6672684},
  organization={IEEE}
}

@inproceedings{gopalakrishnan2012global,
  title={Global optimization of optimal power flow using a branch \& bound algorithm},
  author={Gopalakrishnan, Ajit and Raghunathan, Arvind U and Nikovski, Daniel and Biegler, Lorenz T},
  booktitle={2012 50th Annual Allerton Conference on Communication, Control, and Computing (Allerton)},
  pages={609--616},
  year={2012},
  doi = {10.1109/Allerton.2012.6483274},
  organization={IEEE}
}

@inproceedings{sousa2011heuristic,
  title={A heuristic method based on the branch and cut algorithm to the transmission system expansion planning problem},
  author={Sousa, Aldir S and Asada, Eduardo N},
  booktitle={2011 IEEE Power and Energy Society General Meeting},
  pages={1--6},
  year={2011},
  doi = {10.1109/PES.2011.6039826},
  organization={IEEE}
}

@article{haffner2000branch,
  title={Branch and bound algorithm for transmission system expansion planning using a transportation model},
  author={Haffner, S and Monticelli, A and Garcia, A and Mantovani, J and Romero, R},
  journal={IEE Proceedings-Generation, Transmission and Distribution},
  volume={147},
  number={3},
  pages={149--156},
  year={2000},
  publisher={IET}
}

@inproceedings{megel2016reducing,
  title={Reducing the computational effort of stochastic multi-period DC optimal power flow with storage},
  author={M{\'e}gel, Olivier and Andersson, G{\"o}ran and Mathieu, Johanna L},
  booktitle={2016 Power Systems Computation Conference (PSCC)},
  pages={1--7},
  year={2016},
  doi = {10.1109/PSCC.2016.7541033},
  organization={IEEE}
}

@inproceedings{sahraei2014performance,
  title={Performance of AC and DC based transmission switching heuristics on a large-scale polish system},
  author={Sahraei-Ardakani, Mostafa and Korad, Akshay and Hedman, Kory W and Lipka, Paula and Oren, Shmuel},
  booktitle={2014 IEEE PES General Meeting| Conference \& Exposition},
  pages={1--5},
  year={2014},
  doi = {10.1109/PESGM.2014.6939776},
  organization={IEEE}
}

@article{lorca2016multistage,
  title={Multistage adaptive robust optimization for the unit commitment problem},
  author={Lorca, Alvaro and Sun, X Andy and Litvinov, Eugene and Zheng, Tongxin},
  journal={Operations Research},
  volume={64},
  number={1},
  pages={32--51},
  year={2016},
  doi = {10.1287/opre.2015.1456},
  publisher={INFORMS}
}

@book{villasana1984transmission,
  title={Transmission network planning using linear and linear mixed integer programming},
  author={Villasana, Ramon V},
  address = {United States},
  publisher={Rensselaer Polytechnic Institute},
  year={1984}
}

@inproceedings{tsamasphyrou2000transmission,
  title={Transmission network planning under uncertainty with Benders decomposition},
  author={Tsamasphyrou, Panagiota and Renaud, Arnaud and Carpentier, Pierre},
  booktitle={Optimization: Proceedings of the 9th Belgian-French-German Conference on Optimization Namur, September 7--11, 1998},
  pages={457--472},
  year={2000},
  doi = {10.1007/978-3-642-57014-8_30},
  organization={Springer}
}

@book{wolsey2014integer,
  title={Integer and combinatorial optimization},
  author={Wolsey, Laurence A and Nemhauser, George L},
  address= {USA},
  year={2014},
  isbn = {978-1-118-62686-3},
  publisher={John Wiley \& Sons}
}

@article{madrigal2012transmission,
  title={Transmission expansion for renewable energy scale-up: emerging lessons and recommendations},
  author={Madrigal, Marcelino and Stoft, Steven},
  year={2012},
  doi = {10.1596/978-0-8213-9598-1},
  publisher={World Bank Publications}
}

@article{wan2025grid,
  title={Grid Operational Benefit Analysis of Data Center Spatial Flexibility: Congestion Relief, Renewable Energy Curtailment Reduction, and Cost Saving},
  author={Wan, Haoxiang and Fang, Linhan and Li, Xingpeng},
  journal={arXiv preprint arXiv:2511.08759},
  year={2025}
}

@techreport{ACEG2023ReadyToGo,
  author      = {Zachary Zimmerman and Michael Goggin and Rob Gramlich},
  title       = {Ready-To-Go Transmission Projects 2023: Progress and Status since 2021},
  institution = {Americans for a Clean Energy Grid and Grid Strategies},
  year        = {2023},
  url         = {https://cleanenergygrid.org/wp-content/uploads/2023/09/ACEG_Transmission-Projects-Ready-To-Go_September-2023.pdf}
}

@book{merris2011graph,
  title={Graph theory},
  author={Merris, Russell},
  address= {New York, NY, USA},
  year={2011},
  publisher={John Wiley \& Sons}
}

@article{albert2004structural,
  title={Structural vulnerability of the North American power grid},
  author={Albert, R{\'e}ka and Albert, Istv{\'a}n and Nakarado, Gary L},
  journal={Physical review E},
  volume={69},
  number={2},
  pages={025103},
  doi = {10.1103/PhysRevE.69.025103},
  year={2004},
  publisher={APS}
}

@phdthesis{rahmani2013study,
  title={Study of New Mathematical Models for Transmission Expansion Planning Problem},
  author={Rahmani, M},
  year={2013},
  school={Ilha Solteira UNESP S{\~a}o Paulo, Brazil}
}

@article{lumbreras2014automatic,
  title={Automatic selection of candidate investments for Transmission Expansion Planning},
  author={Lumbreras, Sara and Ramos, Andreas and S{\'a}nchez, Pedro},
  journal={International Journal of Electrical Power \& Energy Systems},
  volume={59},
  pages={130--140},
  doi = {10.1016/j.ijepes.2014.02.016},
  year={2014},
  publisher={Elsevier}
}

@article{bastianel2025identification,
  title={Identification of Technical Design Constraints and Considerations for Transmission Grid Expansion Planning Projects},
  author={Bastianel, Giacomo and Hardy, Clement and Charels, Nils and Van Hertem, Dirk and Ergun, Hakan},
  journal={arXiv preprint arXiv:2512.13496},
  year={2025}
}

@article{fadly2019geographical,
  title={Geographical proximity and renewable energy diffusion: An empirical approach},
  author={Fadly, Dalia and Fontes, Francisco},
  journal={Energy Policy},
  volume={129},
  pages={422--435},
  year={2019},
  publisher={Elsevier}
}

@techreport{Nerc2023,
  author       = {{North American Electric Reliability Corporation}},
  title        = {An Introduction to Inverter-Based Resources on the Bulk Power System},
  institution  = {NERC},
  year         = {2023},
  month        = dec
}

@book{cormen2022introduction,
  title = {Introduction to Algorithms},
  author = {Cormen, Thomas H. and Leiserson, Charles E. and Rivest, Ronald L. and Stein, Clifford},
  year = {2022},
  edition = {4th},
  address = {Cambridge, MA},
  publisher = {MIT Press},
  isbn = {9780262046305}
}

@article{dong2025transmission,
  title={Transmission expansion planning: A deep learning approach},
  author={Dong, Jizhe and Cao, Jianshe and Lu, Yu and Zhang, Yuexin and Li, Jiulong and Xu, Chongshan and Zheng, Danchen and Han, Shunjie},
  journal={Sustainable Energy, Grids and Networks},
  volume={41},
  pages={101585},
  year={2025},
  doi = {10.1016/j.segan.2024.101585},
  publisher={Elsevier}
}

@article{romero2005constructive,
  title={Constructive heuristic algorithm for the DC model in network transmission expansion planning},
  author={Romero, R and Rocha, C and Mantovani, JRS and Sanchez, IG},
  journal={IEE Proceedings-Generation, Transmission and Distribution},
  volume={152},
  number={2},
  pages={277--282},
  year={2005},
  publisher={IET}
}

@article{de2005transmission,
  title={Transmission system expansion planning using a sigmoid function to handle integer investment variables},
  author={de Oliveira, Edimar Jos{\'e} and Da Silva, IC and Pereira, Jos{\'e} Luiz Rezende and Carneiro, S},
  journal={IEEE Transactions on Power Systems},
  volume={20},
  number={3},
  pages={1616--1621},
  year={2005},
  doi = {10.1109/TPWRS.2005.852065},
  publisher={IEEE}
}

@article{hedman2010co,
  title={Co-optimization of generation unit commitment and transmission switching with N-1 reliability},
  author={Hedman, Kory W and Ferris, Michael C and O'Neill, Richard P and Fisher, Emily Bartholomew and Oren, Shmuel S},
  journal={IEEE Transactions on Power Systems},
  volume={25},
  number={2},
  pages={1052--1063},
  doi = {10.1109/TPWRS.2009.2037232},
  year={2010},
  publisher={IEEE}
}

\end{document}